\documentclass{article}

\usepackage[preprint]{neurips_2026}

\usepackage[utf8]{inputenc} % allow utf-8 input
\usepackage[T1]{fontenc}    % use 8-bit T1 fonts
\usepackage{hyperref}       % hyperlinks
\usepackage{url}            % simple URL typesetting
\usepackage{booktabs}       % professional-quality tables
\usepackage{amsfonts}       % blackboard math symbols
\usepackage{nicefrac}       % compact symbols for 1/2, etc.
\usepackage{microtype}      % microtypography
\usepackage{xcolor}         % colors
\usepackage{algorithm}
\usepackage{algpseudocode}
\usepackage{lipsum}
\usepackage{enumitem}
\usepackage{amsmath}
\usepackage{amsthm}
\usepackage{graphicx}
\usepackage{subcaption}
\usepackage{wrapfig}
\usepackage{multirow}

\newtheorem{theorem}{Theorem}[section]
\newtheorem{lemma}{Lemma}[section]
\newtheorem{assumption}{Assumption}[section]
\newtheorem{definition}{Definition}[section]

\newcommand{\cX}{\mathcal{X}}
\newcommand{\cA}{\mathcal{A}}
\newcommand{\cT}{\mathcal{T}}
\newcommand{\norm}[1]{\left\|#1\right\|}
\newcommand{\abs}[1]{\left|#1\right|}

\title{Towards Model-Free Learning in Dynamic Population Games: An Application to Karma Economies}

\author{%
  Matteo Cederle\\
  Department of Information Engineering\\
  University of Padova\\
  \texttt{matteo.cederle@phd.unipd.it} \\
  \And
  Saverio Bolognani \\
  Automatic Control Laboratory \\
  ETH Zurich \\
  \texttt{bsaverio@ethz.ch} \\
  \AND
  Gian Antonio Susto \\
  Department of Information Engineering \\
  University of Padova\\
  \texttt{gianantonio.susto@unipd.it} \\
}

\begin{document}

\maketitle

\begin{abstract}
  Dynamic Population Games (DPGs) provide a tractable framework for modeling strategic interactions in large populations of self-interested agents, and have been successfully applied to the design of Karma economies, a class of fair non-monetary resource allocation mechanisms. Despite their appealing theoretical properties, existing computational tools for DPGs assume full knowledge of the game model and operate in a centralized fashion, limiting their applicability in realistic settings where agents have access only to their own private experience. This paper takes a step towards addressing this gap by studying model-free equilibrium learning in Karma DPGs. First, we analyze the setting in which a novel agent joins a Karma DPG already at its Stationary Nash Equilibrium (SNE) and learns a policy via Deep Q-Networks (DQN) without knowledge of the game model. Leveraging recent convergence results for DQN, we establish a suboptimality bound consisting of a DQN approximation error of order $O(1/\sqrt{N_s})$ and a mean field perturbation error of order $O(1/N)$, where $N_s$ is the replay buffer size and $N$ is the population size. Second, we consider the challenging problem of learning the SNE from scratch. We show empirically that combining deep RL with fictitious play and smoothed policy iteration allows agents to converge, in a model-free fashion, to a configuration close to the centrally computed SNE. Together, these contributions support the vision of Karma economies as practical tools for fair resource allocation.\footnote{Code is provided as a supplementary zip file (see supplementary material). If accepted, it will be made publicly available.}
\end{abstract}

\section{Introduction}
\label{sec:intro}

The allocation of scarce resources among self-interested agents is a fundamental challenge in modern engineered and socio-economic systems. Classical game-theoretic approaches model such interactions as finite $N$-player games, but the analysis and computation of equilibria become intractable as $N$ grows. Mean field games (MFGs), introduced independently by Lasry and Lions \citep{lasry2007mean} and Huang et al. \citep{huang2006large}, overcome this challenge by replacing individual interactions with the interaction of a representative agent with the aggregate behavior of the population, reducing an intractable $N$-player problem to a tractable fixed-point system. This approximation has proven effective across a wide range of applications~\citep{petrakova2022mean, tajeddini2018mean, cardaliaguet2018mean, mao2022mean}, attracting rapidly growing interest from both the control and machine learning communities.

A particularly promising approach to fair scarce resource allocation is represented by \emph{Karma economies}~\citep{elokda2024self, censi2019today}, non-monetary incentive schemes in which agents bid for shared resources using non-tradable tokens that circulate according to designer-specified exchange rules. Karma economies have been shown to achieve fairness and efficiency across multiple allocation problems~\citep{elokda2025carma, cederle2026towards}. Their natural mathematical framework is that of \emph{Dynamic Population Games} (DPGs)~\citep{elokda2024dynamic}, a class of discrete-time, finite state-and-action, stationary MFGs at the intersection of MFGs and classical population games~\citep{sandholm2010population}. DPGs generalize standard MFGs in a key way: both the reward function and transition probabilities depend on the \emph{full} mean field pair $(\mu, \pi)$ --- the joint population state distribution and policy --- rather than on the state distribution alone. This richer coupling, corresponding to the General MFG (GMFG) framework of Guo et al. ~\cite{guo2019learning, guo2023general}, captures the essence of Karma economies, where incentives depend not only on how karma is distributed but also on how agents are bidding. Beyond this modeling flexibility, DPGs enjoy a remarkable structural property: their Stationary Nash Equilibrium (SNE) can be reduced to a Nash Equilibrium of an equivalent static population game, unlocking tools for equilibrium analysis and computation via evolutionary dynamics~\citep{elokda2024dynamic}.

Despite these appealing properties, existing computational tools for DPGs and GMFGs either require strong contractivity assumptions~\citep{guo2019learning, guo2023general, cui2021approximately} rarely verified in practice, or rely on centralized algorithms with full knowledge of the game primitives~\citep{elokda2024dynamic}. In practice, individuals participating in a Karma economy experience only their own private state and receive reward feedback --- they cannot access the game model or observe the population distribution. This raises a natural question: \emph{can agents learn to play a DPG equilibrium in a model-free fashion, using only individual experience?} Reinforcement Learning (RL) offers a natural framework for this problem, and a growing literature has studied learning in MFGs~\citep{guo2019learning, guo2023general, yardim2023policy, subramanian2019reinforcement}. However, as discussed in \autoref{sec:relworks}, model-free learning in the fully coupled DPG setting has received little attention; this paper takes a step towards addressing this gap with the following contributions:

\begin{itemize}[leftmargin=1em]
    \item Motivated by the plug-and-play nature of Karma economies, where new participants can join an established system, we analyze the single-agent setting in which a novel agent joins a Karma DPG already at a SNE configuration and wishes to learn a good policy via Deep Q-Networks (DQN) \citep{mnih2015human} with $\varepsilon$-greedy exploration, without knowledge of the game model. Leveraging recent results on the convergence of DQN \citep{zhang2023convergence}, we establish a suboptimality bound for the learned policy, composed of a DQN approximation error of order $O(1/\sqrt{N_s})$ and a mean field perturbation error of order $O(1/N)$, where $N_s$ is the replay buffer size and $N$ is the population size. Moreover, the population mean field remains within $O(1/N)$ of the equilibrium throughout the learning process, so that the new learning agent does not disrupt the equilibrium of the remaining population. \\
    We validate these guarantees through numerical experiments, where we practically observe the $O(1/\sqrt{N_s})$ and $O(1/N)$ scaling of the two error components for the new agent's policy.
    \item We then consider the challenging problem of learning the Karma DPG equilibrium from scratch, when agents have no information about the game structure. We empirically show that combining deep RL with \textit{fictitious play} \citep{brown1951iterative, cardaliaguet2017learning, perrin2020fictitious} and \textit{smoothed policy iteration} \citep{lauriere2023policy, tang2024learning} allows a population of agents to converge, in a model-free fashion, to a configuration that is close to a centrally computed SNE. While a theoretical analysis of this convergence remains an open problem, our experiments suggest that the combination of deep RL with fictitious play and smoothed policy iteration is a practically viable approach to model-free equilibrium learning in DPGs.
\end{itemize}
 
Taken together, these contributions represent a step towards making DPGs --- and Karma economies in particular --- more accessible as practical tools, by showing that reasonable equilibrium behavior can emerge from model-free learning even when agents have no knowledge of the game model. We view this work as complementary to the existing literature on learning in MFGs and GMFGs, and hope it motivates further investigation into model-free learning methods for this class of games.
\section{Related Works}
\label{sec:relworks}

\paragraph{Learning in Mean Field Games}
A growing body of work has studied how agents can learn equilibria in MFGs through reinforcement learning; we refer the reader to Laurière et al. \cite{lauriere2022learning} for a comprehensive survey. Early works laid the groundwork for RL-based approaches~\citep{guo2019learning, subramanian2019reinforcement, yang2018mean}, but typically assume uniqueness of the Nash equilibrium via contractivity of the mean-field operator~\citep{zeng2024policy, cui2021approximately}, and operate under the classical MFG formulation where rewards and transitions depend only on the state distribution, neither of which holds in Karma economies. The closest work to ours is Guo et al. \cite{guo2023general}, which considers a similar GMFG framework and proves convergence of value-based and policy-based RL algorithms, but still requires stability and uniqueness of the equilibrium. Other policy-based approaches~\citep{xie2021learning, yardim2023policy, ocello2025finite, fouque2025convergence} either impose strong structural assumptions or consider only the classical MFG setting. Notably, Yardim et al.  \cite{yardim2023policy} demonstrate convergence on the true $N$-player game without a population generative model, but do not address the GMFG setting.
From an algorithmic perspective, two main paradigms have emerged: Fictitious Play (FP)~\citep{perrin2020fictitious, elie2020convergence} and Online Mirror Descent (OMD)~\citep{perolat2021scaling, wu2024population}. While OMD enjoys strong theoretical guarantees, its convergence typically requires monotonicity conditions not satisfied in Karma games. FP has proven to be a robust practical alternative~\citep{elie2020convergence, magnino2026bench}, and its combination with deep RL has been explored previously~\citep{lauriere2022scalable, perrin2021mean}, though not in the GMFG framework. Convergence guarantees for deep RL combined with FP remain an open problem even in the classical MFG setting~\citep{lauriere2022scalable}, a limitation we acknowledge; what we show empirically is that this combination is practically effective in the GMFG setting of DPGs, a regime previously unexplored. Rather than proposing a new general framework for learning in MFGs, this work develops theoretical and empirical tools to support agents participating in Karma economies, enabling them to learn effective bidding strategies from individual experience alone.

\paragraph{Mechanisms for fair resource allocation}
Designing allocation mechanisms that are simultaneously fair and efficient is a longstanding challenge. Monetary mechanisms~\citep{vickrey1969congestion, vickrey1961counterspeculation} are the classical approach, but their reliance on money raises significant fairness concerns in socio-technical contexts~\citep{elokda2025vision}, and a fundamental impossibility result~\citep{hylland1980strategy} establishes that no strategy-proof, Pareto-efficient, non-dictatorial rule exists without money in single-shot settings. Non-monetary alternatives~\citep{hylland1979efficient, gao2021online} address these concerns but are typically designed for finite horizons and vulnerable to budget depletion~\citep{elokda2025vision}. Karma economies~\citep{elokda2024self} circumvent these limitations by exploiting the repeated nature of resource allocation: non-tradable tokens that circulate through bidding and redistribution implement Maximum Long-run Nash Welfare in a decentralized and trustworthy manner, without requiring interpersonal utility comparisons~\citep{elokda2025vision}.
\section{Preliminaries}
\label{sec:back}
Let $a,d \in D \subseteq \mathbb{N}$  and let $c \in C\subseteq \mathbb{R}^n$, then for a function $f : D \times C \rightarrow \mathbb{R}$, we distinguish discrete and continuous arguments through the notation $f[d](c)$. Alternatively, we write $f : C \rightarrow \mathbb{R}^{|D|}$ as the vector-valued function $f(c)$, with $f[d](c)$ denoting its $d$-{th} element. Similarly, $p[a \mid d](c)$ denotes the conditional probability of $a$ given $d$ and $c$.
We denote by $\Delta(D):=\{ p \in \mathbb{R}_+^{|D|} \rvert \sum_{d \in D} p[d] = 1\}$ the set of probability distributions over the elements of $D$. For a function $f:D\times C\to\mathbb{R}$, we denote by $||f||_\infty:=\sup_{d\in D,c\in C}|f[d](c)|$ its infinity norm, i.e., the supremum of its absolute value over all discrete and continuous arguments.

\paragraph{Dynamic Population Games (DPGs)}
A DPG, as introduced in Elokda et al. \cite{elokda2024dynamic}, models a large population $\mathcal{N} = \{1, \ldots, N\}$ of anonymous, homogeneous, self-interested agents, where $N$ is large enough that agents approximately form a continuum of mass. Each agent has a private
\emph{state} $x \in \mathcal{X} = \{1, \ldots, N_x\}$, that evolves over time. Both the state and action spaces are finite. At each discrete time step, an agent in state $x$ selects an action $a \in \mathcal{A}[x]$, according to a stationary policy $\pi: \mathcal{X} \to \Delta(\mathcal{A}[x])$, which maps the agent's own state to a probability distribution over actions, without conditioning on the population distribution.
The macroscopic state of the game is described by the \emph{mean field} $(\mu, \pi) \in \mathcal{M} \times \Pi$, where
and $\mu \in \mathcal{M}= \Delta(\mathcal{X})$ is the state distribution of the population. Crucially, both the \emph{state transition function} $p[x^+ \mid x, a](\mu, \pi)$ and the \emph{immediate reward function} $r[x, a](\mu, \pi)$ depend on the \emph{full} mean field pair $(\mu, \pi)$, capturing both the population state distribution and the population behavioral policy. This distinguishes DPGs from classical MFGs, where the coupling is only through the state distribution $\mu$, and places DPGs within the General MFG framework of Guo et al. \cite{guo2019learning}.
 
Each agent faces a discounted infinite-horizon Markov Decision Process (MDP) with discount factor $\alpha \in [0,1)$, aiming to maximize its expected discounted cumulative reward. Each agent’s MDP is coupled to  other agents through the mean field $(\mu, \pi)$. We can define the \emph{state-action value function}:
\begin{equation}
    Q[x, a](\mu, \pi) = r[x, a](\mu, \pi)
    + \alpha \sum_{x^+ \in \cX} p[x^+ \mid x, a](\mu, \pi)\,
      V[x^+](\mu, \pi),
    \label{eq:q_function}
\end{equation}
where $V[x](\mu, \pi) = \sum_{a \in \mathcal{A}[x]} \pi[a|x]Q[x, a](\mu, \pi)$ is the \emph{value function}. The optimal policy for a player in state $x$ belongs to the \emph{best response correspondence} $\mathcal{B}[x](\mu, \pi)$:
\begin{equation}
    \pi^*[\cdot \mid x]\in\mathcal{B}[x](\mu, \pi) = \arg\max_{a \in \mathcal{A}[x]}
    Q[x, a](\mu, \pi).
    \label{eq:best_response}
\end{equation}
 
Finally, we can define the solution concept in Dynamic Population Games.
\begin{definition}[\textit{Stationary Nash Equilibrium \citep{elokda2024dynamic}}]
\label{def:sne}
A mean field $(\mu^*, \pi^*) \in \mathcal{M} \times \Pi$ is a Stationary Nash Equilibrium (SNE) if, for all $x \in \mathcal{X}$:
\begin{align}
    &\mu^*[x] = \sum_{x^- \in \cX} \sum_{a^- \in \mathcal{A}[x^-]} \mu^*[x^-]\,\pi[a^-|x^-]\, 
                   p[x \mid x^-,a^-](\mu^*, \pi^*),
                   \tag{SNE.1} \label{eq:sne1} \\
    &\pi^*[\cdot \mid x] \in \mathcal{B}[x](\mu^*, \pi^*).
                   \tag{SNE.2} \label{eq:sne2}
\end{align}
\end{definition}
 
Condition~\eqref{eq:sne1} requires the state distribution $\mu^*$ to be the stationary distribution of the Markov chain induced by the equilibrium policy $\pi^*$ under the equilibrium mean field, ensuring time-invariance. Condition \eqref{eq:sne2} requires each agent to play optimally given the equilibrium mean field. Together, these conditions define a fixed point that is simultaneously optimal and self-consistent.
 
A key theoretical result of Elokda et al. \cite{elokda2024dynamic} is that every DPG can be reduced to an equivalent \emph{static population game}, whose Nash Equilibria coincide with the SNEs of the DPG. This reduction unlocks a set of tools from the population games literature \citep{sandholm2010population}, including evolutionary dynamics for SNE computation and explicit conditions for existence of the SNE. In particular, under mild continuity assumptions on $r$ and $p$, at least one SNE is guaranteed to exist \citep[Proposition~1]{elokda2024dynamic}.
 
\paragraph{Karma Economies}
Karma economies, introduced in \cite{censi2019today, elokda2024self}, are a class of non-monetary resource allocation mechanisms designed to achieve fairness and efficiency in the repeated allocation of scarce resources among a large population of self-interested agents. The key idea is to issue each agent a budget of non-tradable \emph{karma} tokens that can be used to bid for a shared resource. Karma flows from agents who acquire the resource to those who yield it, according to exchange rules specified by the mechanism designer.
Because karma cannot be traded or accumulated indefinitely, it naturally incentivizes agents to be strategic and forward-looking: spending karma today means having less available for future high-urgency situations. In the Karma DPG model, we take the point of view of an \textit{ego agent} playing against the population, from which a random opponent is drawn in every resource competition instance. The agent has a \textit{private} state $x=[u,k]\in\mathcal{X}=\mathcal{U}\times\{0,...,K\}$, where the urgency state $u$ represents a valuation for the resource and takes one of the values in the discrete and finite set $\mathcal{U}$, and $k$ is the current \textit{karma balance}, limited by some $K\in\mathbb{N}\backslash\{\infty\}$. The ego agent performs actions in the form of
karma bids $a \in \mathcal{A}[k]=\{0,...,k\}$, and it is equipped with the following reward function and state transition probability, dependent on the mean field pair $(\mu,\pi)$:
\begin{align}
    &r[u,k,a](\mu,\pi)=r[u,a](\mu,\pi)=u\sum_{a'}\sum_{u,k}\mu[u,k]\pi[a|u,k]\mathbb{P}[o=\texttt{win}|a,a'], \\
    &p[u^+,k^+|u,k,a](\mu,\pi)=\phi[u^+|u]\kappa[k^+|k,a](\mu,\pi),
\end{align}
where $o\in\mathcal{O}=\{\texttt{win},\texttt{lose}\}$ denotes the resource competition outcome, and $\mathbb{P}[o=\texttt{win}|a,a']=1$ if $a>a'$, $0$ if $a<a'$, and $0.5$ if $a=a'$, with $a'$ being the opponent's action, drawn from the population. Moreover, $\phi[u^+|u]$ is an exogenous, irreducible Markov chain process, and $\kappa[k^+|k,a](\mu,\pi)$ is the \textit{karma transition function}, which can take various forms, depending on the mechanism designer choices; the karma transition function used for the experiments in this work is the \textit{pay bid to society} scheme,
in which the winning agent pays its bid into a common pool that is then redistributed equally among all agents; the full specification of the karma transition kernel is provided in \autoref{appsec:karma_extra}, along with other possible implementations. Finally, as mentioned above for generic DPGs, under mild assumptions, a Karma SNE is guaranteed to exist; however, as mentioned by Elokda et al. \cite{elokda2024self}, the uniqueness of such equilibrium is still an open research question.

\paragraph{Deep Q-Networks}
In high-dimensional, decentralized settings where the game primitives $r$ and $p$ are unknown, the optimal action-value function $Q^*$ cannot be computed analytically and must instead be learned from experience. Deep Q-Networks (DQN) \cite{mnih2015human} is one of the first deep RL algorithms which approximates the optimal Q-function using a deep neural network $Q(\boldsymbol{W}; x, a)$, parameterized by weights $\boldsymbol{W}\in\mathbb{R}^d$, trained to minimize the \emph{mean squared Bellman error} (MSBE):
\begin{equation}
    \min_{\boldsymbol{W}} \; \mathbb{E}_{(x,a) \sim d}\left[
    Q(\boldsymbol{W}; x, a) - r(x,a) - \alpha\,\mathbb{E}_{x'|x,a}[\max_{a'} Q(\boldsymbol{W}; x', a')]
    \right]^2,
    \label{eq:msbe}
\end{equation}
where $d$ is the distribution of $(x,a)$ following the optimal policy $\pi^*$. Two key stabilization techniques distinguish DQN from naive Q-learning with function approximation: an \emph{experience replay buffer} $\mathcal{D}$, which stores past transitions $(x, a, r, x')$ and breaks temporal correlations in the training data by sampling random mini-batches; and a \emph{target network} with frozen weights $\bar{\boldsymbol{W}}$, periodically updated, which provides stable regression targets $y = r + \alpha \max_{a'} Q(\bar{\boldsymbol{W}}; x', a')$ during training. Action selection follows an $\varepsilon$-greedy policy: with probability $\varepsilon$ a random action is chosen for exploration, and with probability $1 - \varepsilon$ the greedy action $\arg\max_a Q(\boldsymbol{W}; x, a)$ is selected. Decaying $\varepsilon$ over time encourages exploration early in training and exploitation once a good policy has been found. Zhang et al. \cite{zhang2023convergence} provide the first theoretical convergence and sample complexity analysis of DQN with $\varepsilon$-greedy exploration, showing that with decaying $\varepsilon$ the learned weights converge geometrically to the optimal weights $\boldsymbol{W}^*$, up to an estimation error of order $O(1/\sqrt{N_s})$, where $N_s$ is the replay buffer size. We build directly on this result in \autoref{sec:method1}.
\section{Decentralized Learning for a Novel Agent in a Karma Economy via DQN}
\label{sec:method1}
We consider a Karma DPG operating at a Stationary Nash Equilibrium $(\mu^*, \pi^*)$. This setting is practically motivated by the deployment nature of Karma economies: in any real-world implementation, new participants continuously join an existing system in which the incumbent population has already learned to play at equilibrium. A Karma economy must ideally be \emph{plug-and-play} in the sense that new agents can join at any time and learn effective bidding strategies from scratch, without disrupting the equilibrium behavior of the existing population. As a first step towards this goal, we analyze the single-agent case, in which one novel non-infinitesimal agent joins the population and seeks to learn a good policy through direct interaction with the game, without any knowledge of the reward function $r$, the transition kernel $p$, or the equilibrium mean field $(\mu^*, \pi^*)$. The agent employs DQN with $\varepsilon$-greedy exploration \citep{mnih2015human} as its learning algorithm. The central question we address is whether such an agent can learn a policy that is close to the equilibrium policy $\pi^*$, and whether its presence disrupts the equilibrium of the remaining population. The new agent faces a \emph{perturbed} MDP: since its policy $\pi_t$ differs from $\pi^*$ during learning, the full mean field shifts from $(\mu^*, \pi^*)$ to a perturbed value $(\tilde{\mu}_t, \tilde{\pi}_t)$, which in turn affects the rewards and transitions experienced by the agent. Our main result, stated in \autoref{thm:main}, shows that the new agent's learned policy is guaranteed to be close to the equilibrium policy, with a suboptimality bound that decomposes into two independent components.
\begin{assumption}
\label{ass:nondegen}
The Stationary Nash Equilibrium $(\mu^*, \pi^*)$ is non-degenerate, in the sense that
the equilibrium Q-function has a strict argmax (action gap) at every state:
\begin{equation}
    \delta := \min_{x \in \mathcal{X}} \left(
        Q^*[x, \pi^*(x)](\mu^*,\pi^*) - \max_{a \neq \pi^*[x]} Q^*[x, a](\mu^*,\pi^*)
    \right) > 0.
    \label{eq:delta}
\end{equation}
\end{assumption}
The concept of action gap plays a central role in approximate reinforcement learning methods, where a larger separation between optimal and suboptimal actions improves the robustness of greedy policies to estimation and approximation errors \cite{bellemare2016increasing}. Moreover, as already shown by \cite[Section 4.2]{guo2023general}, Assumption~\ref{ass:nondegen} should not be viewed as practically restrictive in our scenario, given the compactness of the probability simplex $\mathcal{P}(\mathcal{X}\times\mathcal{A})$ and the finiteness of the action set $\mathcal{A}$.

With this, we can now present a preliminary lemma, followed by the main result of this section, whose proof is provided in \autoref{appsec:proof}.
\begin{lemma}
\label{lem:lipschitz_primitives}
The immediate reward function $r[\cdot](\mu,\pi)$ and the state transition probability
$p[\cdot](\mu,\pi)$ are Lipschitz continuous in $(\mu,\pi)$ on
$\mathcal{M} \times \Pi$. That is, there exist constants $L_r, L_p > 0$ such that
for all $s = (\mu,\pi)$ and $s' = (\mu',\pi')$:
\begin{align*}
    \norm{r[\cdot](s) - r[\cdot](s')}_\infty
    \le L_r \norm{s - s'}, \quad
    \norm{p[\cdot](s) - p[\cdot](s')}_\infty
    \le L_p \norm{s - s'}.
\end{align*}
\end{lemma}
\begin{theorem}
\label{thm:main}
Let Assumption~\ref{ass:nondegen} hold. Consider a Karma DPG
with $N$ agents at a SNE configuration $(\mu^*, \pi^*)$, with population size satisfying: $N > 2C_{MF}/\delta$, where $C_{MF}$ is defined in \eqref{eq:C_MF} and $\delta$ is the action gap from
\eqref{eq:delta}. A new agent joins the population and runs DQN with
$\varepsilon$-greedy exploration \cite{mnih2015human} with replay buffer of size $N_s$
and decaying $\varepsilon_t$. The final learned policy is $\hat{\pi}[x]=\arg\max_{a\in\mathcal{A}[x]}Q(\hat{\boldsymbol{W}};x,a)$, where the learned Q-function $Q(\hat{\boldsymbol{W}};\cdot,\cdot)$ satisfies:
\begin{equation}
    \sup_{(x,a)\in\mathcal{X}\times\mathcal{A}} \abs{Q(\boldsymbol{\hat{W}}; x, a) - Q^*[x, a](\mu^*,\pi^*)}\leq \frac{C_{DQN}}{\sqrt{N_s}} + \frac{C_{MF}}{N},
    %\quad \forall x \in \cX,
    \label{eq:final_bound}
\end{equation}
where $C_{DQN}$ is the DQN approximation constant from Zhang et al. \cite{zhang2023convergence} and:
\begin{equation}
    C_{MF} = \frac{(1 - \alpha) L_r
                   + \alpha R_{\max}|\mathcal{X}| L_p}{(1 - \alpha)^2} \cdot C_0,
    \label{eq:C_MF}
\end{equation}
with $C_0 = \sup_{s,s' \in \mathcal{M} \times \Pi} \norm{s - s'}_\infty$ and $R_{\max}$ the largest immediate reward achievable in the game.
\end{theorem}
The error bound in \eqref{eq:final_bound} admits a natural decomposition into two independent components. The first is the DQN approximation error, of order \(O(1/\sqrt{N_s})\), which decreases with the number of training samples and is independent of the population size \(N\). The second is a mean-field perturbation term, of order \(O(1/N)\), which vanishes as the population grows and is independent of the sample size. 
Finally, we remark that our result does not rely neither on uniqueness of the Stationary Nash Equilibrium nor on global stability conditions on the DPG, which is indeed not guaranteed in the Karma model~\cite{elokda2024self}. 
\section{Towards Model-Free Equilibrium Learning in Karma DPGs}
\label{sec:method2}

We now turn to the more challenging and practically relevant problem of learning the Stationary Nash Equilibrium of a Karma DPG entirely from scratch, without any prior knowledge of the game model or access to the centrally computed equilibrium. In this setting, the agents start from uninformed, arbitrary policies and learn through individual experience.
 
\paragraph{FP-DQN Algorithm}
Our simple approach combines deep reinforcement learning with \emph{fictitious play}~\citep{brown1951iterative, perrin2020fictitious} and \emph{smoothed policy iteration}~\citep{lauriere2023policy, tang2024learning}, two paradigms for equilibrium learning in games in which a representative player best-responds to an empirical average of its own past policies. In DPGs, the analog of the empirical average is the time-averaged mean field $(\bar{\mu}_i, \bar{\pi}_i)$, which aggregates the population distribution and policy across all past iterations. The averaging prevents the algorithm from overreacting to the noisy policy estimates produced by the DQN agent in early iterations, and it ensures that the mean field seen by the representative player evolves smoothly across outer iterations rather than jumping abruptly in response to the most recent learned policy.
 
The full procedure is described in Algorithm~\ref{alg}. At each outer iteration $i$, the representative agent runs DQN for $E$ episodes against a fixed population playing $(\bar{\mu}_i, \bar{\pi}_i)$, the current time-averaged mean field. The population is reset to this configuration at the beginning of each episode, ensuring that the MDP faced by the learning agent is stationary within each outer iteration, consistently with the theoretical setting of \autoref{sec:method1}. After $E$ episodes, the agent's learned policy $\pi_{i,E}$ and the empirical state distribution $\mu_{i,E}$ --- computed at the end of iteration $i$ as the normalized histogram of the population states:
\begin{equation}
    \mu_{i,E}[x] = \frac{1}{N} \sum_{j=1}^{N} \mathbf{1}[x_j = x], \quad \forall x \in \cX,
    \label{eq:empirical_mu}
\end{equation}
where $x_j$ denotes the state of agent $j$ at the end of the final episode --- are recorded and used to update the running averages of the state distribution and population policy, following equation \eqref{eq:fp_update}.

The algorithm is initialized with a uniform policy $\pi_{0,0}$ and state distribution $\mu_{0,0}$, representing a fully uninformed starting point. At convergence, the returned pair $(\bar{\mu}_I, \bar{\pi}_I)$ represents the algorithm's best estimate of the equilibrium policy and state distribution, which can be directly compared against the centrally computed SNE $(\mu^*, \pi^*)$, as we will do in \autoref{subsec:exp2}.
 
\begin{algorithm}
\caption{FP-DQN: Fictitious Play with DQN for Model-Free Equilibrium Learning}
\label{alg}
\begin{algorithmic}[1]
\State Initialize uniform distribution $\bar{\mu}_0 \leftarrow \mu_{0,0}$
\State Initialize uniform policy $\bar{\pi}_0 \leftarrow \pi_{0,0}$
\For{$i = 0, 1, \ldots, I-1$}
    \State Initialize DQN agent with fresh network weights
    \For{$e = 0, 1, \ldots, E-1$}
        \State Reset environment with population playing
               $(\bar{\mu}_i, \bar{\pi}_i)$
        \State Representative agent plays episode $e$ against the population and updates DQN \citep{mnih2015human, zhang2023convergence}
    \EndFor
    \State Compute $\pi_{i,E}$ as the greedy policy of the trained DQN agent
    \State Compute $\mu_{i,E}[x] = \frac{1}{N}\sum_{j=1}^{N} \mathbf{1}[x_j = x]$ for any $x\in\mathcal{X}$ as the empirical population state distribution at the end of episode $E$
    \State Update the state distribution and population policy: \begin{equation}\label{eq:fp_update}
        \bar{\mu}_{i+1} \leftarrow \dfrac{i}{i+1}\,\bar{\mu}_i
           + \dfrac{1}{i+1}\,\mu_{i,E}, \quad \bar{\pi}_{i+1} \leftarrow \dfrac{i}{i+1}\,\bar{\pi}_i
           + \dfrac{1}{i+1}\,\pi_{i,E}
    \end{equation}
\EndFor
\State \textbf{return} $(\bar{\mu}_I,\bar{\pi}_I)$
\end{algorithmic}
\end{algorithm}
 
\paragraph{Discussion and Limitations}
Algorithm~\ref{alg} is best understood as an application-specific tool rather than a general contribution to the MFG learning literature. It combines three well-established techniques --- DQN for model-free policy optimization, fictitious play for mean field averaging, and smoothed policy iteration for population policy averaging --- and applies them to the Karma DPG, a practically motivated instance of the GMFG framework, which has received relatively little attention from the learning community. We therefore do not present Algorithm~\ref{alg} as a general-purpose algorithm for learning in GMFGs. Rather, we view it as a first step towards making Karma economies practically implementable, and as a demonstration that the combination of deep RL and fictitious play is a viable approach to model-free equilibrium learning in the DPG setting.
 
Finally, we acknowledge that a theoretical analysis of the convergence of Algorithm~\ref{alg} remains an open problem. Existing convergence results for model-free learning in GMFGs typically rely on strong structural assumptions that are not satisfied by the Karma DPG, such as the contractivity of the mean field operator and the uniqueness of the equilibrium \citep{guo2019learning, guo2023general}. As discussed in \autoref{sec:back}, the Karma DPG only guarantees existence of at least one SNE under continuity of the game primitives \citep{elokda2024self}, making it a challenging but practically relevant setting for convergence analysis.
\section{Experimental Evaluation}
\label{sec:exp}
The experiments are structured in two parts. In \autoref{subsec:exp1}, we validate the theoretical guarantees of \autoref{thm:main} by analyzing the behavior of a single novel agent
joining a Karma DPG already at its SNE. In \autoref{subsec:exp2}, we
evaluate the FP-DQN algorithm empirically, showing that a population of agents
can converge to a configuration close to the centrally computed SNE starting
from scratch, without any knowledge of the game model. For both parts, full implementation details, including hyperparameter configurations and training schedules, are reported in \autoref{appsec:hyperparams}.

\subsection{Novel Learning Agent in a Karma Economy at Equilibrium}
\label{subsec:exp1}
We consider a simple instance of Karma economy, where the players possess an average amount of karma equal to $\bar{k}=10$, with $K=40$. The agents also have two different urgency levels, i.e., $\mathcal{U}=\{u_1=1,u_2=5\}$, and $\phi[u_2|u_1]=\phi[u_1|u_2]=0.5$. The SNE $(\mu^*, \pi^*)$ is computed centrally using the evolutionary dynamics algorithm of Elokda et al. \cite{elokda2024dynamic},
and serves as the ground truth against which the learned policy is evaluated throughout.
The novel agent is trained with DQN \cite{mnih2015human} with $\varepsilon$-greedy exploration for 1M steps.

\paragraph{Metrics}
To measure the distance between the learned policy $\hat{\pi}$ and the
equilibrium policy $\pi^*$ throughout training, we adopt the Wasserstein-1 distance averaged over
states:
\begin{equation}
    \mathcal{W}(\hat{\pi}, \pi^*) = \frac{1}{|\cX|}
    \sum_{x \in \cX} W_1\bigl(\hat{\pi}(\cdot \mid x),\,
    \pi^*(\cdot \mid x)\bigr),
    \label{eq:wasserstein_metric}
\end{equation}
where $W_1(\cdot, \cdot)$ denotes the Wasserstein-1 \cite{villani2009optimal, lauriere2022learning} distance between two
probability distributions over $\cA$, which measures the minimum cost of transporting one distribution into the other, with cost proportional to the distance between bid values. 
As an additional metric, we also report the \emph{value gap} between the learned policy $\hat{\pi}$ and the equilibrium policy $\pi^*$, estimated by periodically running an evaluation phase during training. Specifically, every 10k training steps, we run 10 evaluation episodes in which the population plays at the SNE $(\mu^*, \pi^*)$. The agent is evaluated twice: once following greedily its current learned policy $\hat{\pi}$, and once following the equilibrium policy $\pi^*$. The value gap is then estimated as the difference in average cumulative undiscounted returns $\widehat{\mathcal{E}}(\hat{\pi})
    = \bar{G}^{\pi^*} - \bar{G}^{\hat{\pi}}$,
where $\bar{G}^{\pi}$ denotes the average cumulative undiscounted return across the 10 evaluation episodes under policy $\pi$, with the population playing at $(\mu^*, \pi^*)$ throughout. By construction, $\widehat{\mathcal{E}}(\hat{\pi}) \geq 0$, with $\widehat{\mathcal{E}}(\hat{\pi}) = 0$ indicating that the agent achieves the same return as an equilibrium player.

\begin{wrapfigure}{r}{0.56\textwidth}
 \centering
    \begin{subfigure}{0.54\textwidth}
        \centering
        \includegraphics[width=\textwidth]{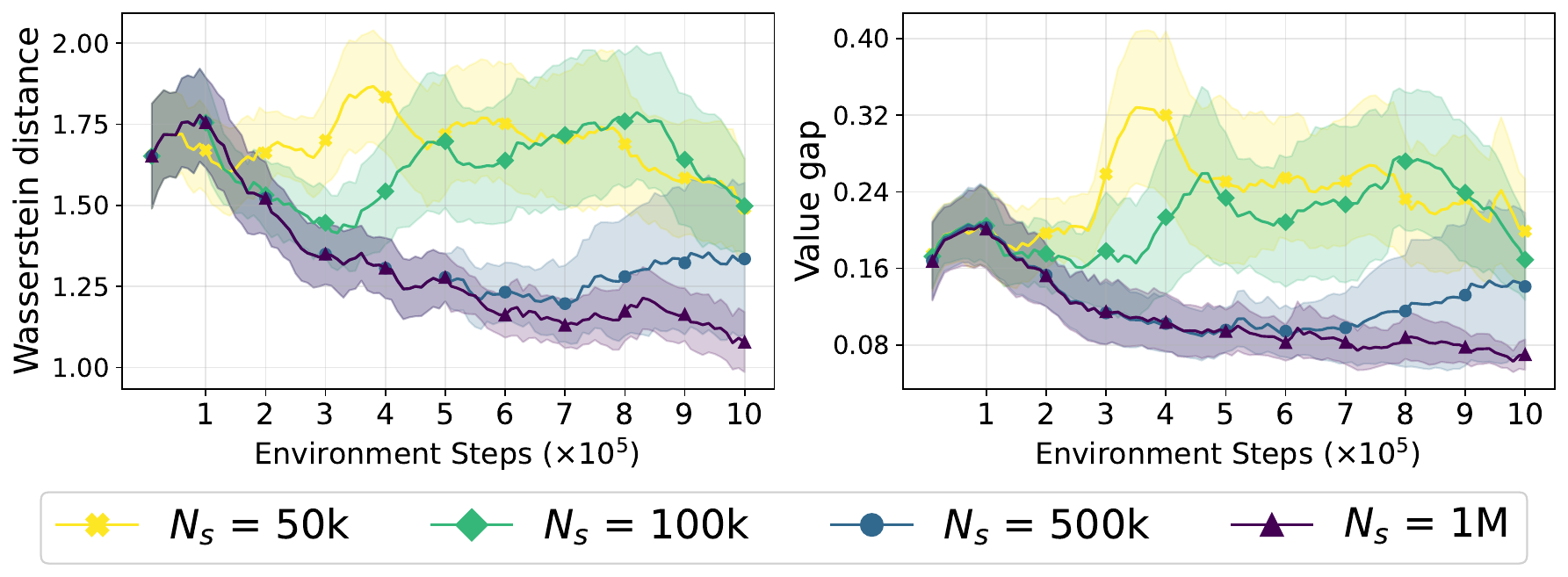}
        \caption{Wasserstein distance and value gap during training for varying buffer sizes $N_s$, with the agent trained against the fixed mean field $(\mu^*, \pi^*)$. Shaded regions: $95\%$ confidence interval across 10 seeds.}
        \label{fig:wass_MF}
    \end{subfigure}
    \begin{subfigure}{0.54\textwidth}
        \centering
        \includegraphics[width=\textwidth]{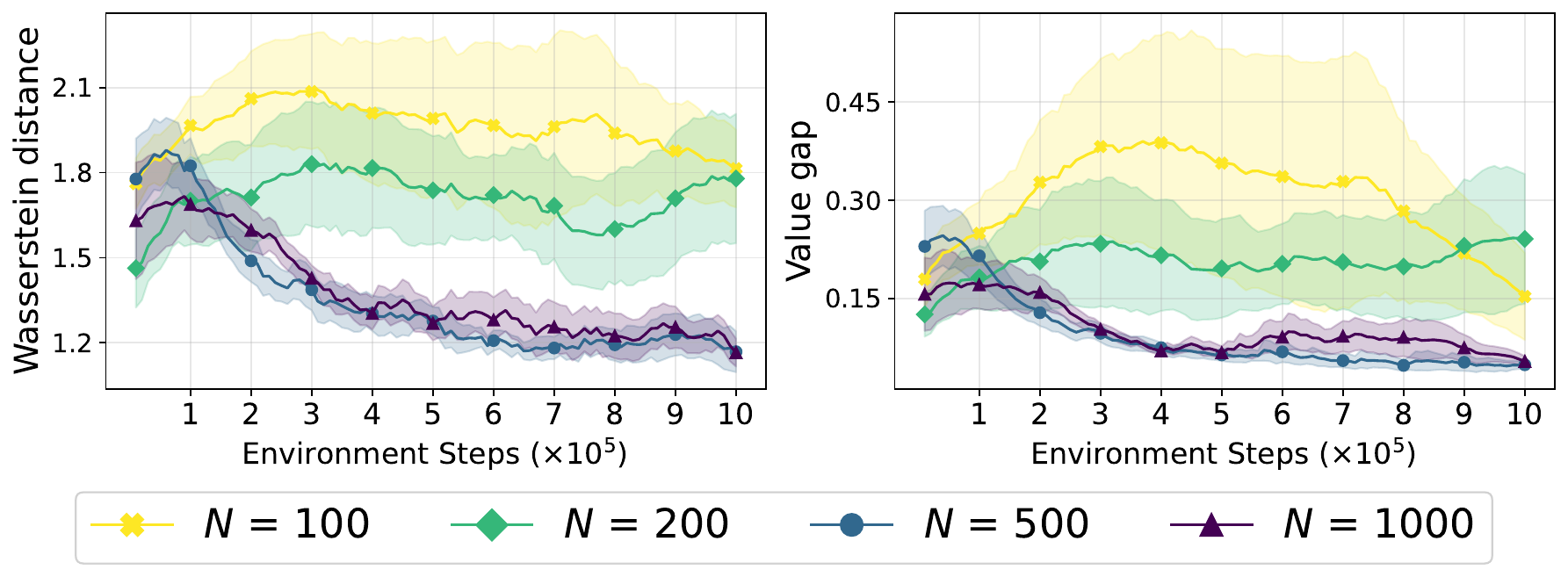}
        \caption{Wasserstein distance and value gap during training for varying population sizes $N$, with fixed $N_s = \text{1M}$. Shaded regions: $95\%$ confidence interval across 10 seeds.}
        \label{fig:wass_PG}
    \end{subfigure}
    
    \caption{Empirical validation of \autoref{thm:main}.}
    \label{fig:wass_combined}
\end{wrapfigure}
\paragraph{Results}
To evaluate separately the two sources of suboptimality identified in \autoref{thm:main}, we conduct two experiments corresponding to the two components of the bound.
In the first experiment, the novel agent is trained directly against the centrally computed mean field $(\mu^*, \pi^*)$, i.e., in the setting described in \autoref{sec:method1}.
This isolates the DQN approximation error by removing any mean field perturbation. We train the agent for $N_s \in \{50\text{k}, 100\text{k}, 500\text{k}, 1\text{M}\}$, and report $\mathcal{W}(\hat{\pi}, \pi^*)$ and $\widehat{\mathcal{E}}(\hat{\pi})$ during training for each configuration in \autoref{fig:wass_MF}.
Both metrics exhibit an overall decreasing trend across training for all buffer size configurations, confirming that the DQN agent progressively learns a policy closer to the equilibrium. As predicted by \autoref{thm:main}, larger replay buffers lead to lower asymptotic error: the agent trained with $N_s = \text{1M}$ achieves the lowest Wasserstein distance and value gap at convergence, while $N_s = \text{50k}$ converges to a noticeably higher error floor.

In the second experiment instead, the novel agent is trained against an actual population of $N$ agents playing at the SNE $(\mu^*, \pi^*)$. This introduces the mean field perturbation effect analyzed in \autoref{sec:method1}.
We fix the replay buffer size to $N_s = \text{1M}$ and vary the population size $N \in \{100, 200, 500, 1000\}$, reporting  $\mathcal{W}(\hat{\pi}, \pi^*)$ and $\widehat{\mathcal{E}}(\hat{\pi})$ during training for each configuration in \autoref{fig:wass_PG}. 
The qualitative behavior mirrors that of the first experiment: both metrics decrease during training for almost all population sizes, with larger $N$ leading to lower asymptotic error, consistently with the $O(1/N)$ mean field perturbation bound of \autoref{thm:main}. As the population grows, the influence of the new agent's non-equilibrium policy on the mean field diminishes, and the MDP it faces becomes progressively closer to the stationary equilibrium MDP of the first experiment. Notably, the agent trained with $N = 1000$ achieves performance comparable to the fixed mean field setting with $N_s = \text{1M}$, suggesting that for sufficiently large populations the mean field perturbation error becomes negligible relative to the DQN approximation error, consistently with \autoref{thm:main}. A visual comparison of the learned and equilibrium policies is reported in \autoref{appsec:add_res}.

\subsection{From scratch Equilibrium Learning in Karma DPGs}
\label{subsec:exp2}
We now evaluate the FP-DQN algorithm, studying whether a population of $N$ agents can converge to a configuration close to the centrally computed SNE starting from a uniform initialization, without any knowledge of the game model. We consider two instances of Karma economy: the instance introduced in \autoref{subsec:exp1}, whose results are reported in \autoref{appsec:add_res} for completeness, and a second instance which we use as the primary setting for the results presented in this section. The players possess an average amount of karma equal to $\bar{k}=10$, with $K=40$, and they have three different urgency levels, i.e., $\mathcal{U}=\{u_1=1,u_2=1,u_3=10\}$. The agents are typically lowly urgent $(u_1 = 1)$, and have a rare occurrence of being highly urgent $(u_3 = 10)$, following
the urgency process $\phi[u^+|u]$, with rows $(0.95, 0.05, 0)$, $(0, 0.5, 0.5)$, and $(0.95, 0.05, 0)$ for $u_1,u_2,u_3$ respectively.
For both instances, the centralized SNE $(\mu^*, \pi^*)$ is computed also in this case using the evolutionary dynamics algorithm of Elokda et al. \cite{elokda2024dynamic}, serving as the ground truth benchmark.
\paragraph{Metrics}
We track across outer iterations $i = 0, 1, \ldots, I-1$ of Algorithm \ref{alg} the
Wasserstein distances $\mathcal{W}(\bar{\pi}_i, \pi^*)$ and
$W_1(\bar{\mu}_i, \mu^*)$ between the current fictitious play iterates and the
centralized Stationary Nash Equilibrium, where $\mathcal{W}(\cdot, \cdot)$ is defined in
\eqref{eq:wasserstein_metric} and $W_1(\cdot, \cdot)$ denotes the
Wasserstein-1 distance between state distributions, as described in \autoref{subsec:exp1}.

\begin{wrapfigure}{r}{0.32\textwidth}
  \centering
  \includegraphics[width=0.3\textwidth]{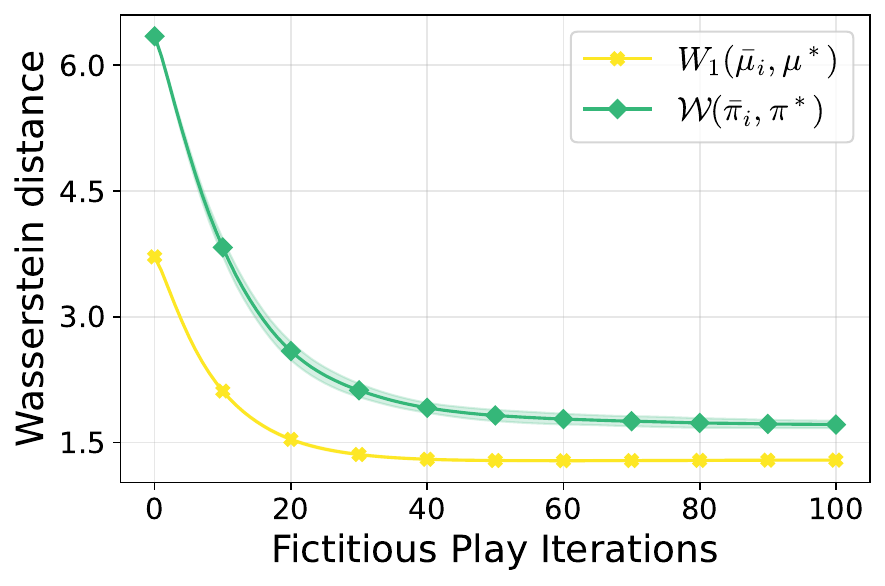}
  \caption{Wasserstein distances between the FP-DQN iterates and the centralized SNE $(\mu^*, \pi^*)$ across fictitious play iterations. Shaded regions represent 95\% confidence interval across 10 random seeds.}
  \label{fig:wass_dist_SNE}
\end{wrapfigure}

\paragraph{Results}
\autoref{fig:wass_dist_SNE} reports the Wasserstein distances between the FP-DQN
iterates $(\bar{\mu}_i, \bar{\pi}_i)$ and the centralized SNE $(\mu^*, \pi^*)$
across training. Both $W_1(\bar{\mu}_i, \mu^*)$ and
$\mathcal{W}(\bar{\pi}_i, \pi^*)$ decrease rapidly and stabilize at low values, indicating that the
algorithm converges to a configuration close to the centrally computed SNE within
a moderate number of iterations. The policy distance $\mathcal{W}(\bar{\pi}_i,
\pi^*)$ starts from a higher value than the distribution distance
$W_1(\bar{\mu}_i, \mu^*)$, which is consistent with the fact that
learning a good policy is a harder task than converging to the correct state
distribution: the distribution is shaped passively by the dynamics once the
policy improves, while the policy itself must be actively learned by the DQN
agent. We remark that
convergence to a configuration close to $(\mu^*, \pi^*)$ was not guaranteed
a priori: in the absence of uniqueness conditions, the
algorithm could in principle have converged to a different fixed point or
failed to converge altogether. Remarkably, agents with no knowledge of the game appear to rediscover the same SNE that a centralized algorithm computes with complete access to the game~primitives.

Finally, \autoref{fig:comparison} presents a direct visual comparison between the centrally computed SNE $(\mu^*, \pi^*)$ and the approximate equilibrium $(\bar{\mu}_I, \bar{\pi}_I)$ returned by FP-DQN at convergence, confirming
that the algorithm correctly captures the qualitative structure of both the equilibrium policy and distribution. 

\begin{figure}[h]
  \centering
  \includegraphics[width=0.9\textwidth]{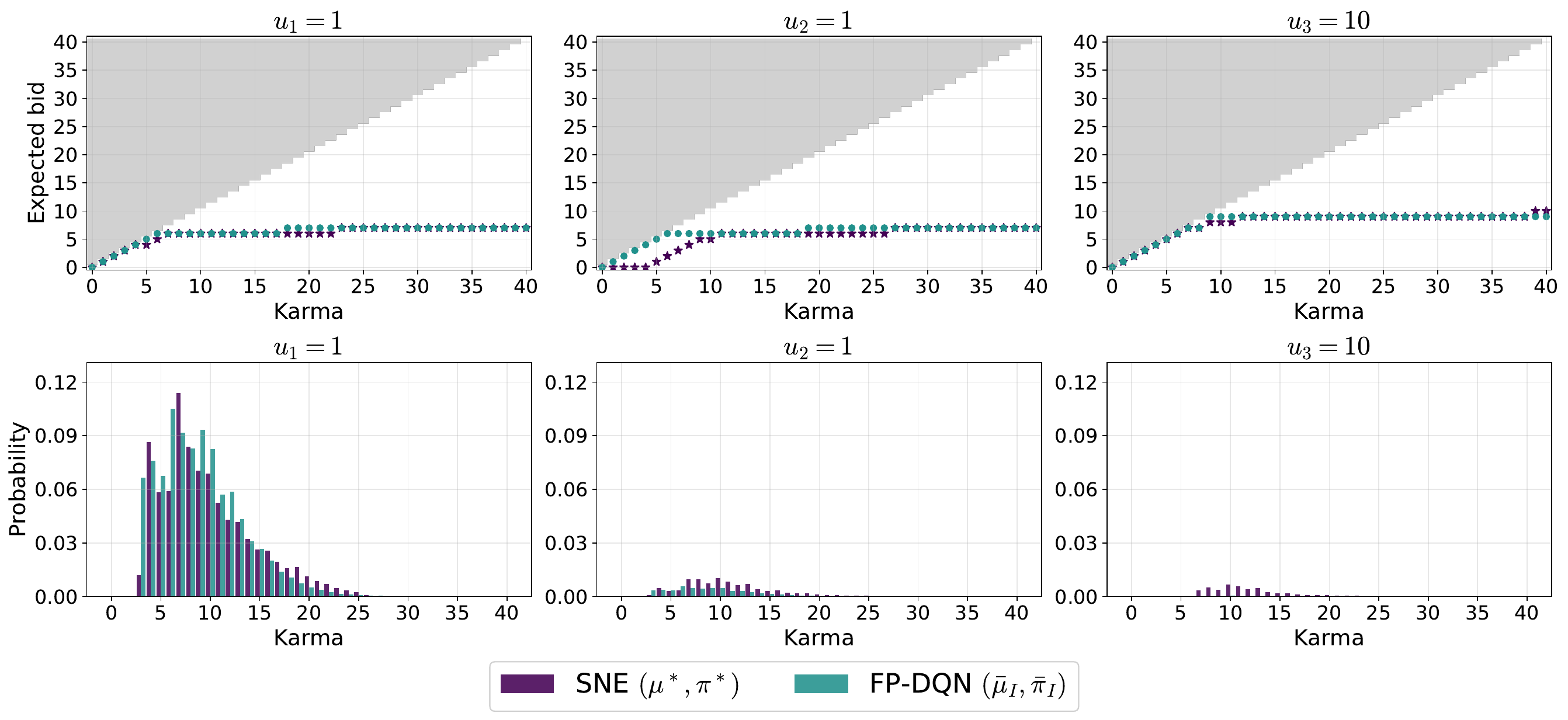}
  \caption{Comparison of the centrally computed SNE $(\mu^*, \pi^*)$ and the approximate equilibrium $(\bar{\mu}_I, \bar{\pi}_I)$ returned by FP-DQN, averaged across 10 seeds. Top: expected bid as a function of karma balance for each urgency level. Bottom: stationary karma distribution for each urgency level.}
  \label{fig:comparison}
\end{figure}
\section{Conclusions}
\label{sec:conc}
This paper studied model-free equilibrium learning in Karma DPGs, a class of GMFGs for fair resource allocation. We provided theoretical guarantees for a novel agent joining a Karma DPG at its SNE and learning via DQN, establishing a suboptimality bound on the learned policy. Additionally, we showed empirically that combining deep RL with FP and smoothed policy iteration allows agents to converge close to the centrally computed SNE without any knowledge of the game model.
\\
We acknowledge that extending convergence guarantees to the from-scratch setting remains an open problem, as well as the uniqueness of the SNE. Additionally, the explicit histogram representation of the population state distribution is feasible in the finite state-action setting of DPGs, but would become intractable in continuous state spaces; extending the framework to such settings, for instance by leveraging generative models \citep{perrin2021mean} is a promising avenue for future work.
\paragraph{Broader Impact} This work contributes to the practical deployment of a class of fair non-monetary resource allocation mechanisms. By enabling agents to learn from individual experience alone, our results bring such mechanisms closer to real-world applicability in domains where fairness and accessibility are critical. We do not foresee significant negative societal impacts from this work.

% \begin{ack}

% \end{ack}

\bibliographystyle{plain}
\bibliography{ref}
%%%%%%%%%%%%%%%%%%%%%%%%%%%%%%%%%%%%%%%%%%%%%%%%%%%%%%%%%%%%

\newpage
\appendix

\section{Appendix}
This appendix provides theoretical derivations, algorithmic details, and extended empirical results to support the main paper. It is organized as follows:
\begin{itemize}[leftmargin=1.5em]
    \item \textbf{\autoref{appsec:karma_extra}}: Karma transition probability
    \item \textbf{\autoref{appsec:proof}}: Complete derivations for Lemma \ref{lem:lipschitz_primitives} and \autoref{thm:main}
    \item \textbf{\autoref{appsec:hyperparams}}: Hyperparameters, experimental details, and hardware configuration
    \item \textbf{\autoref{appsec:add_res}}: Additional results
\end{itemize}

\newpage
\subsection{Karma transition probability}
\label{appsec:karma_extra}
As already mentioned in \autoref{sec:back}, we consider for this work the \textit{pay bid to society} scheme \cite{elokda2024self}, where each agent pays its bid if it is selected, and nothing otherwise. We begin by defining the resource competition outcome probability for the ego agent, given its bid and the full mean field:
\begin{equation}
    \gamma[o|a](\mu,\pi)=\sum_{a'}\sum_{u',k'}\mu[u',k']\pi[a'|u',k']\mathbb{P}[o|a,a'].
\end{equation}
%which leads to $\kappa[k^+|k,a](\mu,\pi)=\sum_o\gamma[o|a](\mu,\pi)\kappa[k^+|k,a,o](\mu,\pi)$
According to this scheme, the conditional payment of the ego agent is $p[a,o]=a$ if $o=\texttt{win}$, and $0$ otherwise. Therefore, the average generated surplus is 
\begin{equation}
    \bar{p}(\mu,\pi)=\sum_{u,k}\mu[u,k]\sum_a\pi[a|u,k]\gamma[o=\texttt{win}|a](\mu,\pi)a,
\end{equation}
which leads to the following redistribution rule: distribute $\lfloor \bar{p}(\mu,\pi)\rfloor$ to a fraction $f^{\text{low}}(\mu,\pi):=\lceil \bar{p}(\mu,\pi)\rceil-\bar{p}(\mu,\pi)$ of randomly selected agents, and $\lceil \bar{p}(\mu,\pi)\rceil$ to the remaining fraction $f^{\text{high}}:=1-f^{\text{low}}$. Finally, we get the following karma transition function:
\begin{align}
\label{eq:k_trans}
    &\kappa[k^+|k,a](\mu,\pi)=\sum_o\gamma[o|a](\mu,\pi)\kappa[k^+|k,a,o](\mu,\pi), \, \text{with} \\
    &\kappa[k^+|k,a,o](\mu,\pi)=\begin{cases}
        f^{\text{low}}(\mu,\pi), \quad &o=\texttt{win}, \, k^+=k-a+ \lfloor \bar{p}(\mu,\pi)\rfloor\\
        f^{\text{high}}(\mu,\pi), \quad &o=\texttt{win}, \, k^+=k-a+ \lceil \bar{p}(\mu,\pi)\rceil\\
        f^{\text{low}}(\mu,\pi), \quad &o=\texttt{lose}, \, k^+=k+ \lfloor \bar{p}(\mu,\pi)\rfloor\\
        f^{\text{high}}(\mu,\pi), \quad &o=\texttt{lose}, \, k^+=k+ \lceil \bar{p}(\mu,\pi)\rceil\\
        0, \quad &\text{otherwise}.
    \end{cases}
\end{align}

\newpage
\subsection{Complete derivations for Lemma \ref{lem:lipschitz_primitives} and \autoref{thm:main}}
\label{appsec:proof}

\begin{proof}
of Lemma \ref{lem:lipschitz_primitives}

\textit{Reward function:}
Recall that in the Karma DPG, the reward function can be written as
\begin{equation}
    r[u,k,a](\mu,\pi)=r[u,a](\mu,\pi)=u\sum_{a'}\sum_{u,k}\mu[u,k]\pi[a|u,k]\mathbb{P}[o=\texttt{win}|a,a'],
\end{equation}
The map $(\mu,\pi) \mapsto r$ is bilinear on the finite-dimensional set
$\mathcal{M} \times \Pi$, hence Lipschitz.

\textit{State transition probability:}
The state transition probability $p[x^+ \mid x,a](\mu,\pi)$ in the Karma DPG is
\begin{equation}    
p[x^+ \mid x,a](\mu,\pi)
= \phi[u^+|u]\kappa[k^+|k,a](\mu,\pi).
\end{equation}
The kernel $\kappa$, defined in \eqref{eq:k_trans},
depends on $(\mu,\pi)$ only through expectations, and is formed through finite sums and piecewise affine
operations. Hence $\kappa$, and therefore also $p$ is Lipschitz on $\mathcal{M} \times \Pi$.

Since $\mathcal{M} \times \Pi$ is compact, all Lipschitz constants can be chosen
uniformly over $(x,a,x^+)$.
\end{proof}

%%%%%%%%%%%%%%%%%%%%%%%%%%%

\begin{proof}
of \autoref{thm:main}

% The proof proceeds in three steps.

% \textit{Step 1: Mean Field Perturbation Bound:}
With $N-1$ agents at equilibrium playing policy $\pi^*$ and one new agent with
current policy $\tilde{\pi}_t$, the empirical population mean field at outer
iteration $t$ of the DQN algorithm \cite{zhang2023convergence} is:
\begin{equation}
    \tilde{s}_t = (\tilde{\mu}_t, \tilde{\pi}_t^{pop})
    = \left(
        \frac{N-1}{N} \mu^* + \frac{1}{N} \tilde{\mu}_t^{agent},\;
        \frac{N-1}{N} \pi^* + \frac{1}{N} \tilde{\pi}_t
      \right),
    \label{eq:perturbed_mf}
\end{equation}
where $\tilde{\mu}_t^{agent}$ is the stationary state distribution induced by
$\tilde{\pi}_t$ on the perturbed MDP. Therefore:
\begin{align}
    \norm{\tilde{s}_t - s^*}
    &= \frac{1}{N}\norm{\left(
        \tilde{\mu}_t^{agent} - \mu^*,\; \tilde{\pi}_t - \pi^*
       \right)}
    \leq \frac{C_0}{N},
    \label{eq:mf_bound}
\end{align}
where $C_0 = \sup_{s,s' \in \mathcal{M} \times \Pi}\norm{s-s'}$ is the diameter of
$\mathcal{M} \times \Pi$, which is finite since both spaces are probability simplices over
finite sets. This bound holds uniformly over all $t \geq 0$ and all possible
policies $\tilde{\pi}_t \in \Pi$, regardless of the learning trajectory of the new
agent.

% \paragraph{Step 2: Lipschitz Bound on the Optimal Q-Function}
% \begin{lemma}[Q-function Lipschitz Continuity]
% \label{lem:q_lipschitz}
Now, we observe that 
under Lemma~\ref{lem:lipschitz_primitives}, for any two mean fields
$s, s' \in \mathcal{M} \times \Pi$:
\begin{equation}
    \norm{Q^*[\cdot, \cdot](s) - Q^*[\cdot, \cdot](s')}_\infty
    \leq L \norm{s - s'},
    \label{eq:q_lip}
\end{equation}
where
\begin{equation}
    L := \frac{(1 - \alpha) L_r
               + \alpha R_{\max}|\mathcal{X}| L_p}{(1 - \alpha)^2}.
    \label{eq:L}
\end{equation}
% \end{lemma}

% \begin{proof}
This is because for fixed $s$, the optimal Q-function $Q^*[\cdot,\cdot](s)$ is the unique fixed
point of the Bellman operator:
\begin{equation}
    (\cT_s Q)[x, a]
    := r[x, a](s)
       + \alpha \sum_{x' \in \cX} p[x' \mid x, a](s)
         \max_{a' \in \cA[x']} Q[x', a'],
    \label{eq:bellman_op}
\end{equation}
which is a contraction with modulus $\alpha \in [0,1)$ in the $\norm{\cdot}_\infty$
norm. Denote $Q = Q^*[\cdot,\cdot](s)$ and $Q' = Q^*[\cdot,\cdot](s')$.
Since both are fixed points of their respective operators:
\begin{equation}
    Q - Q' = \cT_s Q - \cT_{s'} Q'
           = \underbrace{(\cT_s Q - \cT_s Q')}_{(A)}
             + \underbrace{(\cT_s Q' - \cT_{s'} Q')}_{(B)}.
    \label{eq:decomp}
\end{equation}

% \paragraph{Bounding term (A).}
By the contraction property of $\cT_s$:
\begin{equation}
    \norm{(A)}_\infty \leq \alpha \norm{Q - Q'}_\infty.
    \label{eq:bound_A}
\end{equation}

% \paragraph{Bounding term (B).}
Then, expanding using the definition of $\cT_s$ and $\cT_{s'}$:
\begin{align}
    \norm{(B)}_\infty
    &\leq \norm{r[\cdot,\cdot](s) - r[\cdot,\cdot](s')}_\infty
      + \alpha|\mathcal{X}| \norm{p[\cdot|\cdot,\cdot](s)
                         - p[\cdot|\cdot,\cdot](s')}_\infty
        \cdot \norm{V^*[\cdot](s')}_\infty,
    \label{eq:bound_B_expand}
\end{align}
where $V^*[x](s') = \max_{a'} Q'[x,a'](s')$ satisfies
$\norm{V^*[\cdot](s')}_\infty \leq R_{\max}/(1-\alpha)$.
Applying Lemma~\ref{lem:lipschitz_primitives}:
\begin{equation}
    \norm{(B)}_\infty
    \leq \left(L_r + \frac{\alpha R_{\max}|\mathcal{X}|}{1-\alpha} L_p\right)
         \norm{s - s'}.
    \label{eq:bound_B}
\end{equation}

% \paragraph{Combining.}
Now, taking the $\norm{\cdot}_\infty$ norm of \eqref{eq:decomp} and applying
\eqref{eq:bound_A}--\eqref{eq:bound_B}:
\begin{equation}
    \norm{Q - Q'}_\infty
    \leq \alpha \norm{Q - Q'}_\infty
         + \left(L_r + \frac{\alpha R_{\max}|\mathcal{X}|}{1-\alpha}L_p\right)
           \norm{s - s'}.
\end{equation}
Rearranging (valid since $\alpha < 1$) yields \eqref{eq:q_lip} with $L$
as in \eqref{eq:L}.
% \end{proof}

Finally, applying \eqref{eq:q_lip} with $s = \tilde{s}_t$, $s' = s^*$, and
using \eqref{eq:mf_bound}:
\begin{equation}
    \norm{Q^*[\cdot,\cdot](\tilde{s}_t) - Q^*[\cdot,\cdot](s^*)}_\infty
    \leq L \cdot \frac{C_0}{N} = \frac{C_{MF}}{N},
    \label{eq:q_bound}
\end{equation}
where $C_{MF} = L \cdot C_0$ as in \eqref{eq:C_MF}. This bound holds uniformly
in $t$ and independently of the learning trajectory.

% \paragraph{Step 3: DQN Convergence in the Perturbed MDP}
We are now ready to show that the DQN algorithm converges to a policy close to $\pi^*$.
% \paragraph{Structure of the perturbed MDP within each outer iteration.}
A key observation is that within each outer iteration $t$ of Algorithm~1 in
\cite{zhang2023convergence}, the target network weights $\boldsymbol{W}^{(t,0)}$ are fixed. Since
the new agent's policy is determined by $\boldsymbol{W}^{(t,0)}$, the mean field $\tilde{s}_t$
is also fixed within iteration $t$ (cf.\ \eqref{eq:perturbed_mf}). Therefore,
within each outer iteration $t$, the new agent faces a \emph{stationary} MDP
$\mathrm{M}_t$ with fixed mean field $\tilde{s}_t$. The MDP only changes when
moving from outer iteration $t$ to $t+1$, which is exactly the across-iteration
non-stationarity already handled in \cite{zhang2023convergence} through the time-varying
replay buffer $\mathcal{D}_t$ and the distribution shift term $C_t\in[0,1]$, defined by Zhang et al.\ as the fraction of non-optimal state-action pairs $(x,a)$ in the greedy policy with respect to $Q(\boldsymbol{W}^{(t,0)})$.
% \paragraph{Applying Zhang et al.\ to each perturbed MDP.}

Since $\mathrm{M}_t$ is a stationary MDP with finite state-action space,
Assumption 1 in \cite{zhang2023convergence} guarantees the existence of $\boldsymbol{W}^*(\tilde{s}_t)$ achieving
zero MSBE on $\mathrm{M}_t$. The results of \cite{zhang2023convergence} (Theorem~1) therefore
apply directly to $\mathrm{M}_t$ at each outer iteration $t$, giving convergence
of the DQN iterates to $Q^*[\cdot,\cdot](\tilde{s}_t)$ up to estimation error
$O(1/\sqrt{N_s})$.
% \paragraph{Handling the shift in $C_t$ across outer iterations.}

We recall that the term $C_t$ in \cite{zhang2023convergence} measures the fraction of state-action pairs
where the current greedy policy differs from the optimal policy. In our setting,
the optimal policy of $\mathrm{M}_t$ is $\pi^*_{\tilde{s}_t}$, which potentially
differs from $\pi^*_{s^*}$ because the MDP changes across iterations. We now show
that under Assumption~\ref{ass:nondegen} and  the condition $N > 2C_{MF}/\delta$,
this difference is in fact zero.

The optimal policies $\pi^*_{\tilde{s}_t}$ and $\pi^*_{s^*}$ can only disagree at
states $x$ where the perturbation is large enough to flip the argmax, i.e.\ where
there exist actions $a_1 \neq a_2$ such that:
\begin{equation}
    Q^*[x, a_1](s^*) \geq Q^*[x, a_2](s^*)
    \quad \text{but} \quad
    Q^*[x, a_1](\tilde{s}_t) < Q^*[x, a_2](\tilde{s}_t).
    \label{eq:flip_condition}
\end{equation}
For \eqref{eq:flip_condition} to hold, the Q-function gap between $a_1$ and $a_2$
under $s^*$ must satisfy:
\begin{equation}
    Q^*[x, a_1](s^*) - Q^*[x, a_2](s^*)
    < 2\norm{Q^*[\cdot,\cdot](\tilde{s}_t) - Q^*[\cdot,\cdot](s^*)}_\infty
    \leq \frac{2C_{MF}}{N},
    \label{eq:gap_condition}
\end{equation}
where the last inequality uses \eqref{eq:q_bound}. However, by
Assumption~\ref{ass:nondegen}, the minimum action gap at the equilibrium satisfies:
\begin{equation}
    Q^*[x, \pi^*(x)](s^*) - \max_{a \neq \pi^*(x)} Q^*[x, a](s^*) \geq \delta
    \quad \forall x \in \cX.
    \label{eq:gap_lower}
\end{equation}
Recall that we have $2C_{MF}/N < \delta$, so
\eqref{eq:gap_condition} cannot be satisfied for any $x \in \cX$. Therefore:
\begin{equation}
    \pi^*_{\tilde{s}_t}[x] = \pi^*_{s^*}[x] \quad \forall x \in \cX,\; \forall t,
    \label{eq:policy_agreement}
\end{equation}
and hence $C_t$ in our setting coincides exactly with $C_t$ in \cite{zhang2023convergence}.
The across-iteration non-stationarity introduces no additional error beyond what is
already accounted for in the analysis of \cite{zhang2023convergence}.
% \paragraph{Convergence to the perturbed optimal Q-function.}

By the above, the conditions of Theorem~1 in \cite{zhang2023convergence} are fully satisfied
in our setting at each outer iteration. At convergence (i.e.\ for
$T \geq \log_\rho(1/N_s)$ outer iterations, where
$\rho = \alpha + c_\varepsilon(1-\alpha) < 1$, $c_\varepsilon$ defined as in \cite{zhang2023convergence}), the learned weights $\boldsymbol{\hat{W}}$
satisfy:
\begin{equation}
    \sup_{x,a} \abs{Q(\boldsymbol{\hat{W}}; x, a) - Q^*[x, a](\tilde{s}_T)}
    \leq \frac{C_{DQN}}{\sqrt{N_s}},
    \label{eq:zhang_result}
\end{equation}
where $\tilde{s}_T$ is the mean field at the final outer iteration $T$.

% \paragraph{Transferring convergence to the equilibrium Q-function.}
Applying the triangle inequality to \eqref{eq:zhang_result} and \eqref{eq:q_bound}:
\begin{align}
    \sup_{x,a} \abs{Q(\boldsymbol{\hat{W}}; x, a) - Q^*[x, a](s^*)}
    &\leq \sup_{x,a} \abs{Q(\boldsymbol{\hat{W}}; x, a) - Q^*[x, a](\tilde{s}_T)}
          \notag \\
    &\quad + \norm{Q^*[\cdot,\cdot](\tilde{s}_T) - Q^*[\cdot,\cdot](s^*)}_\infty
          \notag \\
    &\leq \frac{C_{DQN}}{\sqrt{N_s}} + \frac{C_{MF}}{N},
    \label{eq:q_total_bound}
\end{align}

% \paragraph{From Q-function error to policy suboptimality.}
which completes the proof.

\end{proof}

\newpage
\subsection{Hyperparameters, experimental details, and hardware configuration}
\label{appsec:hyperparams}
All the models were trained on a server equipped with 2 AMD EPYC 9224 CPUs (48 cores, 96 threads), 1.5 TB of RAM, and 6 NVIDIA L40S GPUs (46 GB VRAM each). All the results presented in the paper are averaged across 10 independent
random seeds, with 95\% confidence intervals computed via the standard error of the mean.
For all the experiments, we consider the discount factor $\alpha=0.98$, and episodes composed by 1000 interaction steps.

\paragraph{DQN architecture and training}
The Q-network is a fully connected neural network with 2 hidden
layers of 64 neurons each, with ReLU activations. The input to the
network is the state representation $x = [u, k] \in \cX$, encoded as a one-hot vector for $u$, concatenated with a scalar equal to $k$ normalized by $K$. The output is a vector of $|\cA|=K+1$ values, one per action.
The network is trained using the RMSprop optimizer with
learning rate $\eta = 0.0003$. The target network is updated every 1000 gradient steps. Mini-batches of size 128 are sampled uniformly from the
replay buffer at each gradient step. The $\varepsilon$-greedy schedule decays
geometrically from $\varepsilon_0 = 1$ to $\varepsilon_{\min} = 0.01$ over 1M steps.
Similarly to Zhang et al. \cite{zhang2023convergence}, we consider the Double DQN architecture \cite{van2016deep}. Finally, action masking is applied at each step to enforce the karma feasibility constraint $a \leq k$, ensuring the agent never selects a bid exceeding its current balance.

\paragraph{FP-DQN hyperparameters.}
The FP-DQN algorithm is run for $I = 100$ outer iterations, each consisting
of $E = $100 episodes of 1000 steps each. The DQN agent is re-initialized
with fresh network weights at the start of each outer iteration. The empirical
state distribution $\mu_{i,E}$ is computed as the normalized histogram of
population states at the end of episode $E$, with population size $N = $1000.

\newpage
\subsection{Additional results}
\label{appsec:add_res}

\paragraph{Novel Learning Agent in a Karma Economy at Equilibrium}
Figures~\ref{fig:policy_0}- \ref{fig:policy_3} report the expected bid as a function of karma balance for the centrally computed SNE policy $\pi^*$ and the learned policies $\hat{\pi}$ obtained in all the tested configurations of the two experiments of \autoref{subsec:exp1}. Considering in particular \autoref{fig:policy_3}, we observe that both learned policies closely recover the qualitative structure of the equilibrium. These observations are consistent with the quantitative results of \autoref{fig:wass_combined}: the policy learned with $N_s = \text{1M}$ is slightly closer to $\pi^*$ than the one learned with $N = 1000$, reflecting the additional $O(1/N)$ mean field perturbation error present in the second experiment.
%Figures \ref{fig:policy_0}-\ref{fig:policy_2} show a similar qualitative behavior, with deviations from $\pi^*$ increasing as $N_s$ decreases or $N$ decreases.

\begin{figure}[h]
  \centering
  \includegraphics[width=\textwidth]{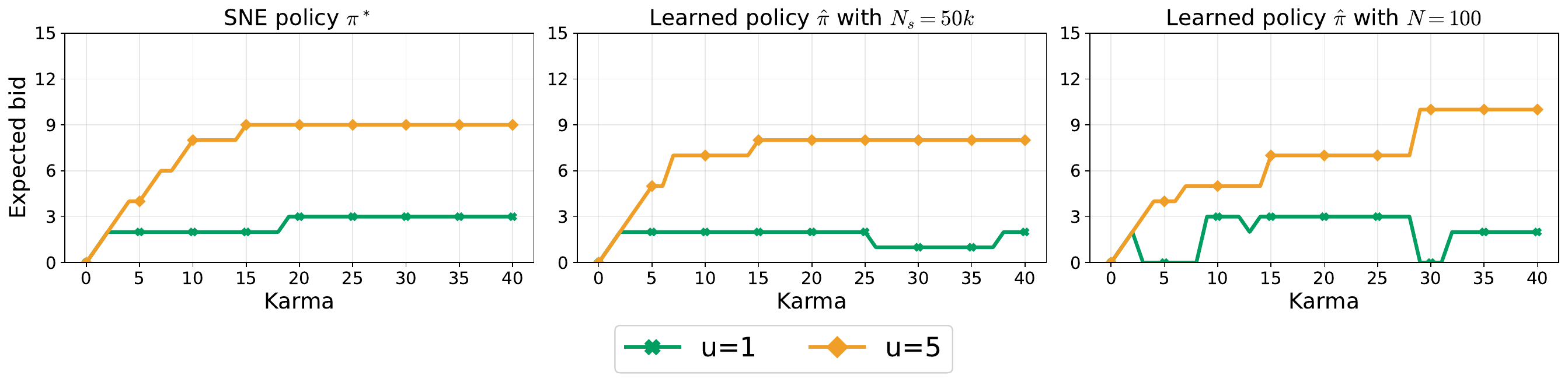}
  \caption{Comparison of the centrally computed equilibrium policy $\pi^*$ and the learned policies by setting $N_s=50k$ and $N=100$, respectively. The results are averaged across 10 seeds.}
  \label{fig:policy_0}
\end{figure}

\begin{figure}[h]
  \centering
  \includegraphics[width=\textwidth]{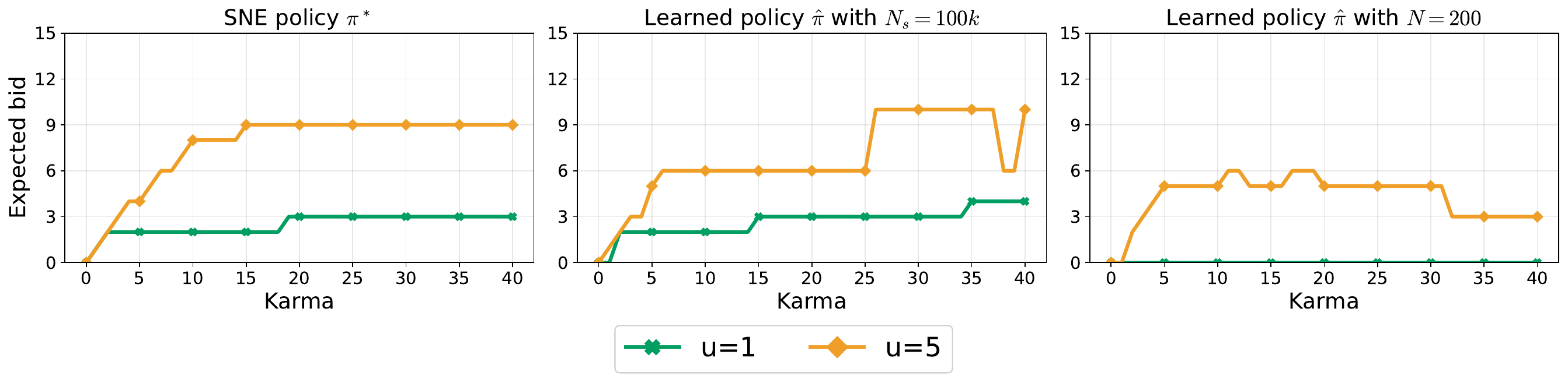}
  \caption{Comparison of the centrally computed equilibrium policy $\pi^*$ and the learned policies by setting $N_s=100k$ and $N=200$, respectively. The results are averaged across 10 seeds.}
  \label{fig:policy_1}
\end{figure}

\begin{figure}[h]
  \centering
  \includegraphics[width=\textwidth]{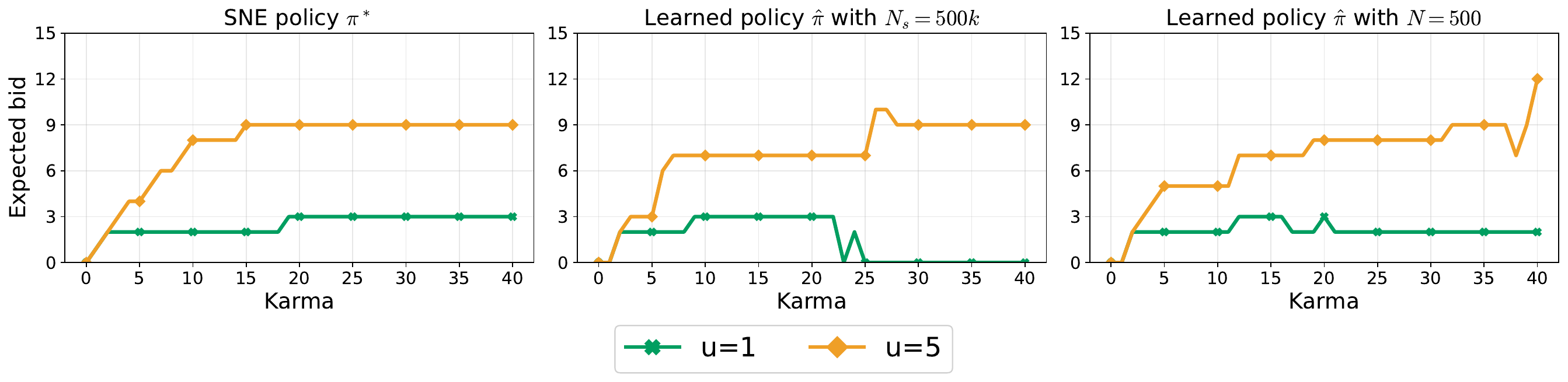}
  \caption{Comparison of the centrally computed equilibrium policy $\pi^*$ and the learned policies by setting $N_s=500k$ and $N=500$, respectively. The results are averaged across 10 seeds.}
  \label{fig:policy_2}
\end{figure}

\clearpage
\begin{figure}[h]
  \centering
  \includegraphics[width=\textwidth]{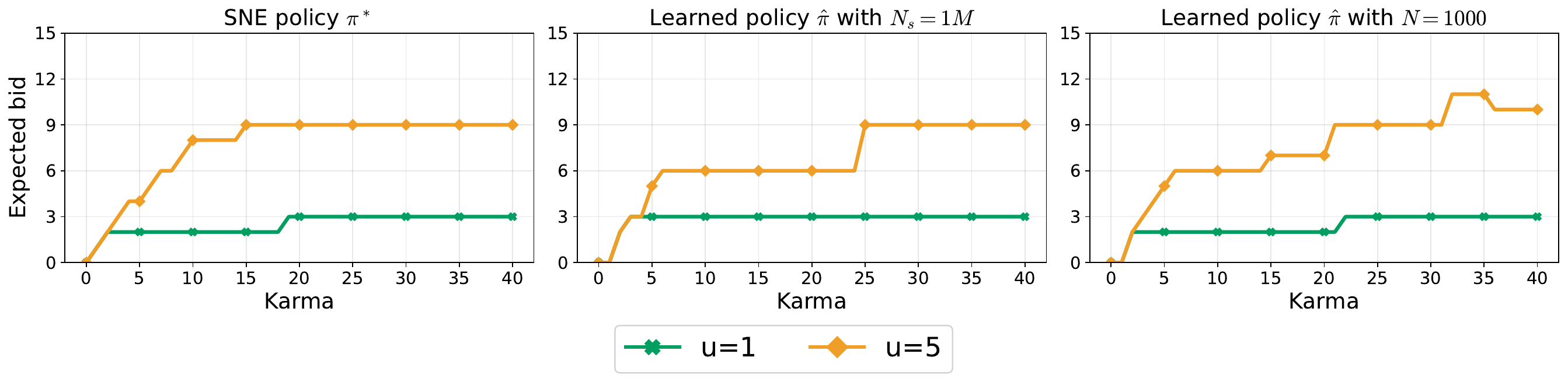}
  \caption{Comparison of the centrally computed equilibrium policy $\pi^*$ and the learned policies by setting $N_s=1M$ and $N=1000$, respectively. The results are averaged across 10 seeds.}
  \label{fig:policy_3}
\end{figure}

\begin{wrapfigure}{r}{0.42\textwidth}
    \centering
    \includegraphics[width=0.4\textwidth]{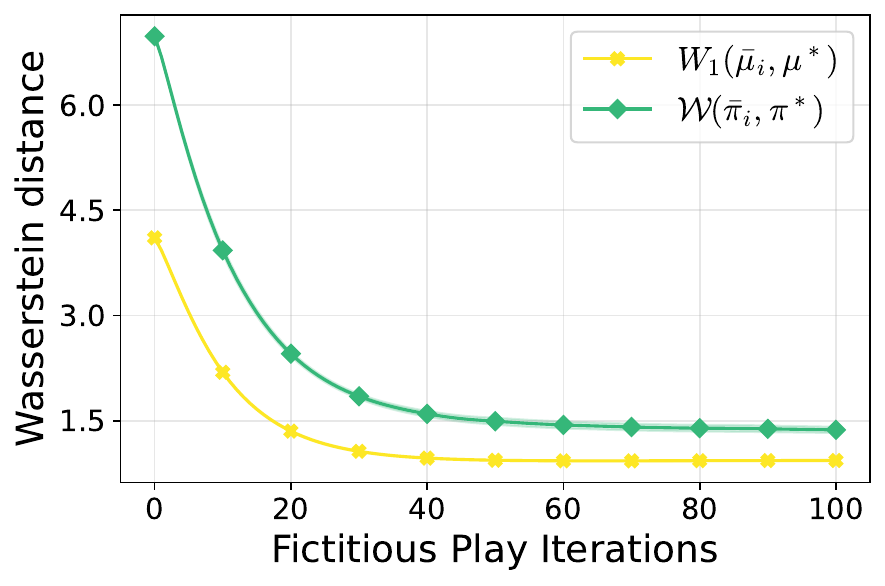}
    \caption{Wasserstein distances between the FP-DQN iterates and the centralized SNE $(\mu^*, \pi^*)$ across fictitious play iterations, for the first Karma economy instance. Shaded regions represent 95\% confidence interval across 10 random seeds.}
    \label{fig:wass_dist_SNE_instance1}
\end{wrapfigure}
\paragraph{From scratch Equilibrium Learning in Karma DPGs}
We report here the results of the FP-DQN algorithm on the first Karma economy instance introduced in \autoref{subsec:exp1}, with two urgency levels $\mathcal{U} = \{u_1 = 1, u_2 = 5\}$ and $\phi[u_2|u_1] = \phi[u_1|u_2] = 0.5$. \autoref{fig:wass_dist_SNE_instance1} reports the Wasserstein distances between the FP-DQN iterates and the centralized SNE across fictitious play iterations, and \autoref{fig:comparison_instance1} shows the visual comparison between the centrally computed SNE and the approximate equilibrium returned by FP-DQN at convergence. The results are qualitatively consistent with those reported in \autoref{subsec:exp2} for the second instance: both metrics decrease and stabilize at low values within a moderate number of outer iterations, and the learned policy and distribution closely match the structure of the equilibrium. This confirms that the convergence behavior of FP-DQN is robust across different Karma economy configurations and is not specific to a single parameter regime.

\begin{figure}[h]
    \centering
    \includegraphics[width=0.85\textwidth]{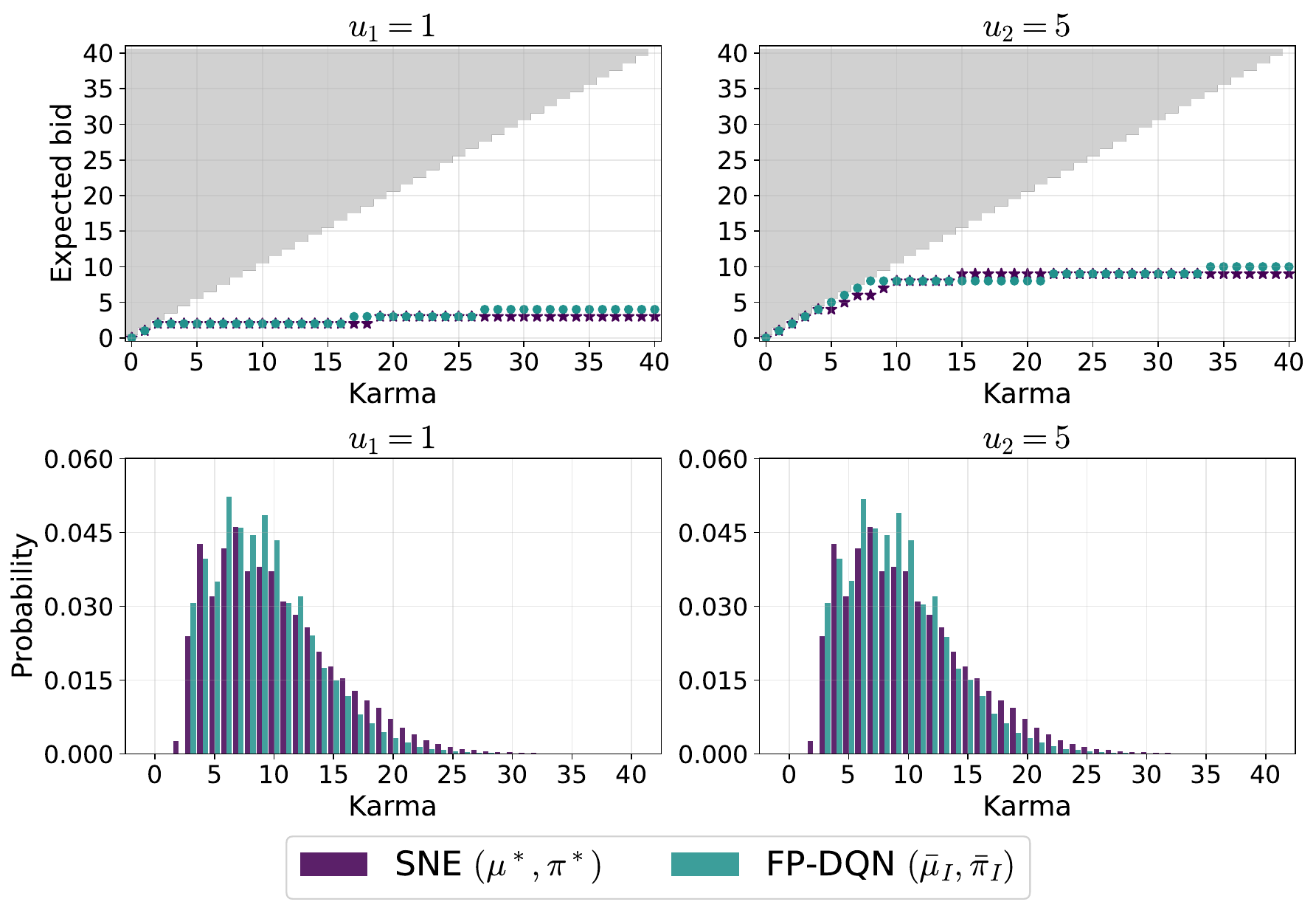}
    \caption{Comparison of the centrally computed SNE $(\mu^*, \pi^*)$ and the approximate equilibrium $(\bar{\mu}_I, \bar{\pi}_I)$ returned by FP-DQN at convergence, averaged across 10 seeds, for the first Karma economy instance. Top: expected bid as a function of karma balance for each urgency level. Bottom: stationary karma distribution for each urgency level.}
    \label{fig:comparison_instance1}
\end{figure}
%%%%%%%%%%%%%%%%%%%%%%%%%%%%%%%%%%%%%%%%%%%%%%%%%%%%%%%%%%%%

\clearpage
\section*{NeurIPS Paper Checklist}

\begin{enumerate}

\item {\bf Claims}
    \item[] Question: Do the main claims made in the abstract and introduction accurately reflect the paper's contributions and scope?
    \item[] Answer: \answerYes{} % Replace by \answerYes{}, \answerNo{}, or \answerNA{}.
    \item[] Justification: The abstract and introduction accurately reflect the paper's contributions and scope: theoretical guarantees are explicitly scoped to the single-agent setting, and the from-scratch convergence result is consistently presented as empirical.

\item {\bf Limitations}
    \item[] Question: Does the paper discuss the limitations of the work performed by the authors?
    \item[] Answer: \answerYes{} % Replace by \answerYes{}, \answerNo{}, or \answerNA{}.
    \item[] Justification:  Both in the conclusions and in the sections presenting the new methods, limitations are highlighted, including the restriction of theoretical guarantees to the single-agent setting, and the scalability of the framework to continuous state spaces.

\item {\bf Theory assumptions and proofs}
    \item[] Question: For each theoretical result, does the paper provide the full set of assumptions and a complete (and correct) proof?
    \item[] Answer: \answerYes{} % Replace by \answerYes{}, \answerNo{}, or \answerNA{}.
    \item[] Justification:  In the appendix, we provide formal proofs for each theoretical result in the main text, with numbered equations.

    \item {\bf Experimental result reproducibility}
    \item[] Question: Does the paper fully disclose all the information needed to reproduce the main experimental results of the paper to the extent that it affects the main claims and/or conclusions of the paper (regardless of whether the code and data are provided or not)?
    \item[] Answer: \answerYes{} % Replace by \answerYes{}, \answerNo{}, or \answerNA{}.
    \item[] Justification: Both in the main text and in the appendix, we provide details on the setup, hyperparameters, and hardware being used. Furthermore, for easy better interpretation, we include in the main text a pseudocode of the from-scratch learning algorithm. Finally, we provide the code as supplementary material.

\item {\bf Open access to data and code}
    \item[] Question: Does the paper provide open access to the data and code, with sufficient instructions to faithfully reproduce the main experimental results, as described in supplemental material?
    \item[] Answer: \answerYes{} % Replace by \answerYes{}, \answerNo{}, or \answerNA{}.
    \item[] Justification: The code for the proposed methods is included in the supplementary material and publicly available. All data can be generated during training.

\item {\bf Experimental setting/details}
    \item[] Question: Does the paper specify all the training and test details (e.g., data splits, hyperparameters, how they were chosen, type of optimizer) necessary to understand the results?
    \item[] Answer: \answerYes{} % Replace by \answerYes{}, \answerNo{}, or \answerNA{}.
    \item[] Justification: We list important hyperparameters, neural network architectures, and other training details in the appendix

\item {\bf Experiment statistical significance}
    \item[] Question: Does the paper report error bars suitably and correctly defined or other appropriate information about the statistical significance of the experiments?
    \item[] Answer: \answerYes{} % Replace by \answerYes{}, \answerNo{}, or \answerNA{}.
    \item[] Justification: We assume normally distributed errors. All results are averaged over 10 runs and reported with 95\% confidence interval.

\item {\bf Experiments compute resources}
    \item[] Question: For each experiment, does the paper provide sufficient information on the computer resources (type of compute workers, memory, time of execution) needed to reproduce the experiments?
    \item[] Answer: \answerYes{} % Replace by \answerYes{}, \answerNo{}, or \answerNA{}.
    \item[] Justification: In the appendix, the hardware specifications and usages are specified.
    
\item {\bf Code of ethics}
    \item[] Question: Does the research conducted in the paper conform, in every respect, with the NeurIPS Code of Ethics \url{https://neurips.cc/public/EthicsGuidelines}?
    \item[] Answer: \answerYes{} % Replace by \answerYes{}, \answerNo{}, or \answerNA{}.
    \item[] Justification: Our paper conforms to the NeurIPS Code of Ethics.

\item {\bf Broader impacts}
    \item[] Question: Does the paper discuss both potential positive societal impacts and negative societal impacts of the work performed?
    \item[] Answer: \answerYes{} % Replace by \answerYes{}, \answerNo{}, or \answerNA{}.
    \item[] Justification: We discuss them at the end of the main paper.
    
\item {\bf Safeguards}
    \item[] Question: Does the paper describe safeguards that have been put in place for responsible release of data or models that have a high risk for misuse (e.g., pre-trained language models, image generators, or scraped datasets)?
    \item[] Answer: \answerNA{} % Replace by \answerYes{}, \answerNo{}, or \answerNA{}.
    \item[] Justification: We find that this paper poses no such risk.

\item {\bf Licenses for existing assets}
    \item[] Question: Are the creators or original owners of assets (e.g., code, data, models), used in the paper, properly credited and are the license and terms of use explicitly mentioned and properly respected?
    \item[] Answer: \answerYes{} % Replace by \answerYes{}, \answerNo{}, or \answerNA{}.
    \item[] Justification: We implemented most of our algorithms from scratch. All our results are generated during training.

\item {\bf New assets}
    \item[] Question: Are new assets introduced in the paper well documented and is the documentation provided alongside the assets?
    \item[] Answer: \answerYes{} % Replace by \answerYes{}, \answerNo{}, or \answerNA{}.
    \item[] Justification: We provide our code and instructions to run the experiments.

\item {\bf Crowdsourcing and research with human subjects}
    \item[] Question: For crowdsourcing experiments and research with human subjects, does the paper include the full text of instructions given to participants and screenshots, if applicable, as well as details about compensation (if any)? 
    \item[] Answer: \answerNA{} % Replace by \answerYes{}, \answerNo{}, or \answerNA{}.
    \item[] Justification: The paper does not involve crowdsourcing nor research with human subjects.

\item {\bf Institutional review board (IRB) approvals or equivalent for research with human subjects}
    \item[] Question: Does the paper describe potential risks incurred by study participants, whether such risks were disclosed to the subjects, and whether Institutional Review Board (IRB) approvals (or an equivalent approval/review based on the requirements of your country or institution) were obtained?
    \item[] Answer: \answerNA{} % Replace by \answerYes{}, \answerNo{}, or \answerNA{}.
    \item[] Justification: The paper does not involve crowdsourcing nor research with human subjects.

\item {\bf Declaration of LLM usage}
    \item[] Question: Does the paper describe the usage of LLMs if it is an important, original, or non-standard component of the core methods in this research? Note that if the LLM is used only for writing, editing, or formatting purposes and does \emph{not} impact the core methodology, scientific rigor, or originality of the research, declaration is not required.
    %this research? 
    \item[] Answer: \answerNA{} % Replace by \answerYes{}, \answerNo{}, or \answerNA{}.
    \item[] Justification: LLMs have been used only for writing, editing, and formatting purposes.

\end{enumerate}

\end{document}